\documentclass[lettersize,journal]{IEEEtran}
\usepackage{amsmath,amssymb}

\usepackage{subfigure}
\usepackage{bm}
\usepackage{graphicx,graphics,color,psfrag}
\usepackage{cite,balance}
\usepackage{caption} 
\allowdisplaybreaks
\usepackage{algorithm}
\usepackage{algorithmic}
\usepackage{accents}
\usepackage{amsthm}
\usepackage{url}
\usepackage[english]{babel}
\usepackage{multirow}
\usepackage{enumerate}
\usepackage{cases}
\usepackage{stfloats}
\usepackage{dsfont}
\usepackage{color,soul}
\usepackage{amsfonts}
\usepackage{cite,graphicx,amsmath,amssymb}
\usepackage{subfigure}
\usepackage{fancyhdr}
\usepackage{hhline}
\usepackage{graphicx,graphics}
\usepackage{array,color}
\usepackage{framed} 
\usepackage{mathtools}
\usepackage{tcolorbox}
\usepackage{amsmath}
\usepackage{geometry}

\newtheorem{definition}{\emph{\underline{Definition}}}

\newtheorem{lemma}{\emph{\underline{Lemma}}}

\newtheorem{remark}{\bf \emph{\underline{Remark}}}

\def\({\left(}
\def\){\right)}

\def\b0{{\mathbf{0}}}

\begin{document}
\captionsetup[figure]{name={Fig.},labelsep=period,singlelinecheck=off}  
\title{Rotatable Antennas for Integrated Sensing and Communications}
\author{Chao Zhou, Changsheng You, Beixiong Zheng, Xiaodan Shao, and Rui Zhang,~\IEEEmembership{Fellow,~IEEE}
\thanks{C. Zhou and C. You  are with the Department of Electronic and Electrical Engineering, Southern University of Science and Technology (SUSTech), Shenzhen
518055, China (e-mail: zhouchao2024@mail.sustech.edu.cn, youcs@sustech.edu.cn). 
B. Zheng is with School of Microelectronics, South China University of Technology, Guangzhou 511442, China (e-mail: bxzheng@scut.edu.cn).
X. Shao is with the Department of Electrical and Computer Engineering, University of Waterloo, Waterloo, ON N2L 3G1, Canada (e-mail: x6shao@uwaterloo.ca).
R. Zhang is with the School of Science and Engineering, Shenzhen
Research Institute of Big Data, The Chinese University of Hong Kong,
Shenzhen, Guangdong 518172, China (e-mail: rzhang@cuhk.edu.cn).
\emph{(Corresponding author: Changsheng You.)}  
}}

\maketitle
	\begin{abstract}
 	 In this letter, we propose to deploy rotatable antennas (RAs) at the base station (BS) to enhance both communication and sensing (C\&S) performances, by exploiting a new spatial degree-of-freedom (DoF) offered by array rotation. Specifically, we formulate a multi-objective optimization problem to simultaneously maximize the sum-rate of multiple communication users and minimize the Cramér-Rao bound (CRB) for target angle estimation, by jointly optimizing the transmit beamforming vectors and the array rotation angle at the BS. 
 	 To solve this problem, we first equivalently decompose it into two subproblems, corresponding to an inner problem for beamforming optimization and an outer problem for array rotation optimization.
 	 Although these two subproblems are non-convex, we obtain their high-quality solutions by applying the block coordinate descent (BCD) technique and one-dimensional exhaustive search, respectively.
 	 Moreover, we show that for the communication-only case, RAs provide an additional rotation gain to improve communication performance; while for the sensing-only case, the equivalent spatial aperture can be enlarged by RAs for achieving higher sensing accuracy. 
 	 Finally, numerical results are presented to showcase the performance gains of RAs over  fixed-rotation antennas in integrated sensing and communications (ISAC). 
\end{abstract}

 \begin{IEEEkeywords}
		Rotatable antenna, integrated sensing and communication, Cram\' er-Rao bound.
\end{IEEEkeywords}

\section{Introduction}

	As the key enabling technique for the sixth-generation (6G) wireless network, integrated sensing and communications (ISAC) is expected to achieve both ubiquitous connectivity and pervasive sensing cost-effectively~\cite{LiuCRB2022,Shao2022Sensing}. 
	To this end, one efficient approach is to jointly design the signal waveforms for simultaneously enhancing the communication and sensing (C\&S) performances. This approach, however, may not be effective in practice when communication users and sensing targets are located in orthogonal subspaces~\cite{LiuCRB2022}.
	
	To exploit new spatial degrees of freedom (DoFs) for enhancing the C\&S performances,
	movable antennas (MAs)~\cite{Zhu2024Tutorial} and fluid antenna systems (FAS)~\cite{YouNGAT,WongFAS}, have been recently proposed, which achieve high resolution and beamforming gains via flexible adjustment of the antennas' positions without
	increasing the number of antennas~\cite{ma2024movable,zhou2024fluid}. However, the array rotation gain has not been well exploited in MAs/FAS, thereby limiting their performance in C\&S.
	To tackle this issue, the six-dimensional MA (6DMA) has recently been proposed to effectively exploit the DoF provided by antennas/subarrays' rotation for improving the wireless network capacity and sensing accuracy~\cite{Shao6DMATWC,Shao2025JSAC}. By jointly optimizing the three-dimensional (3D) positions and 3D rotations of all 6DMA subarrays in a given site space, the 6DMA-equipped transmitter/receiver can adaptively allocate antenna resources to maximize the array gain and spatial multiplexing gain according to the users'/targets' spatial distribution. As a special case of 6DMA, rotatable antennas (RAs) with fixed positions, have recently gained growing attention~\cite{zheng2025rotatable,Li2024RA} due to their low implementation cost and complexity. Specifically, RAs can reshape the antenna radiation pattern in the angular domain for achieving flexible beam coverage by adjusting the rotation angles of antennas. 
	While some initial works~(e.g., \cite{zheng2025rotatable,Li2024RA}) have studied the RA-based communication systems in enhancing the received power and transmission coverage, the potential of RAs for ISAC systems remains largely unexplored. For example, it remains unknown how much performance gain can be obtained by RAs for both C\&S.
	
	Motivated by the above, we consider in this letter a new RA-enabled ISAC system and study the C\&S performance gains provided by array rotation. Specifically, we consider the scenario where an RA-array is deployed at a base station (BS) to serve multiple communication users in the downlink and sense a target via echo signals in the uplink, as shown in Fig.~\ref{Fig:systemmodel_mo}.
	An optimization problem is formulated to simultaneously maximize the sum-rate of communication users and minimize the Cramér-Rao bound (CRB) for target angle estimation by jointly designing the transmit beamforming and the array rotation at the BS. To solve this non-convex problem, we first equivalently decompose it into two subproblems, corresponding to an inner problem for optimizing  beamforming vectors and an outer problem for designing rotation angle of the BS.
	Then, the block coordinate descent (BCD) technique and one-dimensional (1D) exhaustive search are leveraged to solve these two subproblems, respectively. Furthermore,  we show that for the communication-only case, RAs can enhance the received signal-to-noise ratio (SNR) at the communication user; while for the sensing-only case, the equivalent spatial aperture can be enlarged  by RAs for achieving more accurate target sensing. Finally, numerical results demonstrate the performance gains of RAs over  fixed-rotation antennas in improving both C\&S performances.

\section{System Model}\label{Sec2:label}
We consider an RA-enabled ISAC system as shown in Fig.~\ref{Fig:systemmodel_mo}, where a uniform linear RA-array equipped with $M_{\rm t}$ transmit antennas and $M_{\rm r}$ receive antennas is deployed at the BS to serve  $K$ single-antenna communication users in the downlink and sense one target via echo signals simultaneously in the monostatic ISAC mode.\footnote{The proposed design can be extended to the multi-target case by addressing additional issue of rank deficiency~\cite{LiuCRB2022}, which is left for our future work.} In addition to the transmit beamforming design, the 1D rotation angle of RA-array is judiciously controlled to balance between C\&S performances.

\begin{figure}[t]
	\centering
	\includegraphics[width=0.4\textwidth]{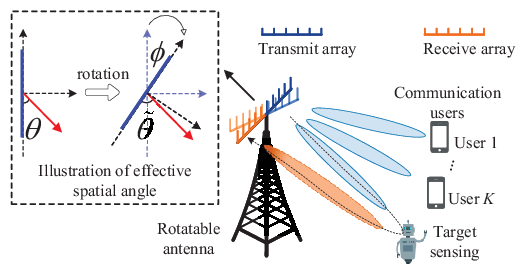}
	\caption{The schematic of RA-enabled ISAC systems.} \label{Fig:systemmodel_mo}
	\vspace{-10pt}
\end{figure}
\subsection{Model for Communication Users}
Let $ \mathbf{h}^H_{k} $ denote the channel from the RA-array to user $k$, which is modeled based on  far-field planar wavefronts \cite{LiuCRB2022}:
\begin{align}
	\mathbf{h}^H_{k}  = \beta_{k,0} \mathbf{a}^H(\tilde{\theta}_{k,0}  ) +  \sum_{\ell=1}^{L} \beta_{k,\ell} \mathbf{a}^H( \tilde{\theta}_{k,\ell}),
\end{align}
where $\beta_{k,0}$ and $\beta_{k,\ell}$ represent the complex path gains of the line-of-sight (LoS) path and the $\ell$-th non-LoS (NLoS) path for user $k$, respectively.\footnote{The omni-directional antenna gain model is adopted in this work for simplicity. However, the proposed scheme can be extended to the general directional antenna gain model considered in~\cite{Shao6DMATWC,Shao2025JSAC,zheng2025rotatable}.} $ \mathbf{a}({\tilde{\theta}_{k,\ell}}) $ denotes the channel steering vector accounting for the \emph{array-wise} antenna rotation, where $ \theta_{k,\ell} $ is the spatial angle,  $\phi$ is the RA-array rotation angle, and $\tilde{\theta}_{k,\ell} = \theta_{k,\ell} + \phi$, $\ell \in \{0,1,\ldots,L\}$, represents the \emph{effective} spatial angle of the $\ell$-th path accounting for the array's rotation. As such,  $\mathbf{a}({\tilde{\theta}_{k,\ell}})\in \mathbb{C}^{  M_{\rm t} \times 1  }$ can be modeled~as
\begin{align}
	\mathbf{a}({\tilde{\theta}_{k,\ell}}) = [e^{j \frac{2\pi}{\lambda} d_{{\rm t},1} \sin\tilde{\theta}_{k,\ell} },\ldots, e^{j \frac{2\pi}{\lambda} d_{{\rm t},M_{\rm t}} \sin\tilde{\theta}_{k,\ell} }]^T,
\end{align}
where $ d_{{\rm t},m_{\rm t}} $ denotes the distance between the $m_{\rm t}$-th transmit antenna and the center of the transmit array. For the uniform linear RA-array, given the half-wavelength inter-antenna spacing $d = \frac{\lambda}{2}$, $ d_{{\rm t},m_{\rm t}} $ can be obtained as
\begin{align}\label{Sum_dis_tr}
	d_{{\rm t},m_{\rm t}} =\frac{(m_{\rm t}-1)D_{\rm t}}{M_{\rm t}-1}  - \frac{D_{\rm t}}{2}, \forall m_{\rm t}\in \{1,2,\ldots,M_{\rm t}\},
\end{align}
where $ D_{\rm t} = (M_{\rm t} - 1)d $  is the aperture of the transmit array.

In the monostatic ISAC mode, the RA-array transmits $K$ independent information signals to $K$ users, which are reflected by the target as echoes for detecting the target at the BS. Let $T$ denote the length of transmission frame with $t \in \{1,2,\ldots,T\}$ representing its $t$-th slot. As such, the transmitted signal at the $t$-th slot, denoted by $ \mathbf{x}_{\rm t} \in \mathbb{C}^{M_{\rm t} \times 1 }$, is given by
$	\mathbf{x}_{\rm t} =\sum_{k=1}^{K}\mathbf{w}_{k} s_{k}$,
where $ \mathbf{w}_{k}\in \mathbb{C}^{M_{\rm t} \times 1} $  denotes the transmit beamforming vector for user $k$ and $s_{k} \sim \mathcal{CN}(0,1) $ denotes its transmitted symbol. Based on the above, the signal-to-interference-plus-noise ratio (SINR) at the $k$-th user can be expressed as 
\begin{align}
	\gamma_{k} = \frac{ |\mathbf{h}_{k}^{H} \mathbf{w}_{k}|^2 }{\sum_{i=1,i\neq k}^{K} |\mathbf{h}_{k}^{H} \mathbf{w}_{i}|^2 +\sigma_{k}^{2} },\forall k\in\{1,\ldots,K\},
\end{align}
with $\sigma_{k}^2$ denoting the received power of the additive white Gaussian noise (AWGN).

\subsection{Model for Target Sensing}
For point target sensing, we assume that there exists an LoS channel between the RA-array and the target, which has been widely adopted in radar sensing~\cite{LiuCRB2022}. Let  $ \mathbf{H}_{\rm s} $ denote the round-trip channel matrix from the transmit array to the target and then to the receive array, which can be modeled as
$	\mathbf{H}_{\rm s} = \beta_{\rm s} \mathbf{a}^H_{\rm T}(\tilde{\theta}) \mathbf{a}_{\rm R}(\tilde{\theta}).$
Herein, $ \beta_{\rm s} $ denotes the round-trip complex path gain accounting for radar cross section (RCS),  $ \mathbf{a}_{\rm T}(\tilde{\theta}) $ and $ \mathbf{a}_{\rm R}(\tilde{\theta}) $ represent the transmit and receive steering vectors, respectively,  which can be  modeled as
\begin{align}
		\mathbf{a}_{\rm T}({\tilde{\theta}}) = [e^{j \frac{2\pi}{\lambda} d_{{\rm t},1} \sin\tilde{\theta} },\ldots, e^{j \frac{2\pi}{\lambda} d_{{\rm t},M_{\rm t}}\sin\tilde{\theta} }]^T,\\
		\mathbf{a}_{\rm R}({\tilde{\theta}}) = [e^{j \frac{2\pi}{\lambda} d_{{\rm r},1} \sin\tilde{\theta} },\ldots, e^{j \frac{2\pi}{\lambda} d_{{\rm r},M_{\rm r}}\sin\tilde{\theta} }]^T,
\end{align}
where $d_{{\rm r},m_{{\rm r}}}$ denotes the distance between the $m_{{\rm r}}$-th receive antenna and the center of the receive array.
Note that $\tilde{\theta} =  \theta + \phi$ represents the \emph{effective} spatial angle of the target accounting for the array's rotation of the BS.
 
During the frame transmission period, the received signal matrix at the receive array, denoted by $ \mathbf{Y}_{\rm s} \in \mathbb{C}^{M_{\rm r} \times T }$, is given by 
	$\mathbf{Y}_{\rm s} = \mathbf{H}_{\rm s} \mathbf{X} + \mathbf{Z}_{\rm s}$,
where $\mathbf{X} = [\mathbf{x}_{1} ,\ldots,\mathbf{x}_{T} ] \in \mathbb{C}^{M_{\rm t} \times T} $ and $ \mathbf{Z}_{\rm s}\in \mathbb{C}^{M_{\rm r} \times T } $ is the noise matrix with each entry following $\mathcal{CN}(0,\sigma_{\rm s}^2)$. 

\subsection{Performance Metrics}\label{Sec2:C}
\textbf{\emph{Communication Metric}}: For communication users, we consider the achievable sum-rate as the communication performance metric, which is given by
\begin{align}
		f_{\rm c} \left( \{\mathbf{w}_{k}\},\phi \right) &=  \sum_{k=1}^{K}\log_2(1+\gamma_{k}).
\end{align}
Since the SINR for user $k$ is determined by both the transmit beamforming vectors $\{\mathbf{w}_{k}\}$ and the array rotation angle $\phi$, adjusting the rotation angle of RA-array offers an additional spatial DoF for improving the communication performance, which will be investigated in Section~\ref{Sec3-D}. 

\textbf{\emph{Sensing Metric}}: For target sensing, we consider CRB for angle estimation of $\tilde{\theta}$, which characterizes the lower bound of unbiased estimation for angle $\tilde{\theta}$.
Note that $\tilde{\theta}$ is a direction of interest for target sensing, which is typical in target tracking scenarios under which the BS designs its beamforming and array rotation towards an estimated/predicted angle~\cite{LiuCRB2022}.
According to~\cite{LiuCRB2022}, the CRB for estimation $\tilde{\theta}$ can be expressed as
\begin{align}\label{CRB_generalform}
	\text{CRB}\big( \tilde{\theta} \big)  = \frac{1}{2  T \text{SNR}_{\rm s} \big|\big|\dot{\mathbf{a}}_{\rm R}(\tilde{\theta}) \big|\big|^2 \big(  \mathbf{a}_{\rm T}^H(\tilde{\theta}) \mathbf{W} \mathbf{a}_{\rm T}(\tilde{\theta}) \big)  },
\end{align}
where $ \mathbf{W}  = \sum_{k=1}^{K}  \mathbf{w}_{k}  \mathbf{w}_{k}^{H}$, $\text{SNR}_{\rm s} =\frac{|\beta_{\rm s} |^2}{\sigma_{\rm s}^2} $, and  $ \dot{\mathbf{a}}_{\rm R}  $ represents the derivative of the receive steering vector ${\mathbf{a}}_{\rm R}$ with respect to (w.r.t.) $\tilde{\theta}$, which is given by~\cite{LiuCRB2022} 
\begin{align}\label{der_steer}
	\dot{\mathbf{a}}_{\rm R} = j \frac{2\pi}{\lambda} \cos\tilde{\theta} \Big[d_{1}  e^{j \frac{2\pi}{\lambda} d_{{\rm r},1} \sin\tilde{\theta} },\ldots,d_{M_{\rm r}}e^{j \frac{2\pi}{\lambda} d_{{\rm r},M_{\rm r}}\sin\tilde{\theta} }\Big]^T.
\end{align}

\begin{lemma}\label{lemma_CRB}
	\rm 
	Under the uniform linear array model of~\eqref{Sum_dis_tr}, the CRB for target angle estimation $\tilde{\theta}$ can be rewritten as
	\begin{align}\label{CRB}
		\text{CRB}\big( \tilde{\theta} \big)  = \frac{\chi}{\cos^2\tilde{\theta} \left(  \mathbf{a}_{\rm T}^H(\tilde{\theta}) \mathbf{W} \mathbf{a}_{\rm T}(\tilde{\theta})\right) },
	\end{align}
where $\chi = \frac{ 3 \lambda^2 ( M_{\rm r}-1)}{2 \pi^2 T \text{SNR}_{\rm s} M_{\rm r} (M_{\rm r} + 1) D_{\rm r}^2 }$, with $D_{\rm r} = (M_{\rm r} - 1)d $ representing the aperture of the receive array.
\end{lemma}
\begin{proof}
	Substituting $d_{{\rm r},m_{\rm r}}\!=\!\frac{(m_{\rm r}-1)D_{\rm r}}{M_{\rm r}-1} \!-\! \frac{D_{\rm r}}{2}, \forall  m_{\rm r}\!\in\!\{\!1,2,\ldots,M_{\rm r}\!\}$ and~\eqref{der_steer}  into~\eqref{CRB_generalform}, we can obtain the CRB for $\tilde{\theta}$ in~\eqref{CRB}.
\end{proof}
It is observed from~\eqref{CRB} that minimizing the angle estimation CRB is equivalent to maximizing
\begin{align}\label{perso-CRB}
	f_{\rm s} \left(\{\mathbf{w}_{k}\},\phi \right) &=  \cos^2\tilde{\theta} \left(  \mathbf{a}_{\rm T}^H(\tilde{\theta}) \mathbf{W} \mathbf{a}_{\rm T}(\tilde{\theta}) \right),
\end{align}
which is jointly determined by the transmit beamforming vectors \{$\mathbf{w}_{k}$\} and the array rotation angle $\phi$.

\section{Problem Formulation and Proposed Solution}\label{Sec:3}
In this section, we first formulate an optimization problem to optimize both C\&S performances. To solve this non-convex problem, an efficient algorithm is proposed to obtain its suboptimal solution by using the BCD technique and 1D exhaustive search.
In addition, we analytically show that for the communication-only case, array rotation provides an additional gain for enhancing the received SNR, while for the sensing-only case, equivalent spatial aperture can be enlarged for enhancing sensing accuracy by optimizing the array rotation angle.

\subsection{Problem Formulation}
Our target is to jointly maximize the achievable sum-rate of communication users and minimize the CRB for target angle estimation (which is equivalent to maximizing~\eqref{perso-CRB}). This problem can be mathematically formulated as 
\begin{subequations}
	\begin{align}
		(\textbf{P1}):\; \max_{\{\mathbf{w}_{k}\},\phi}~& \varpi_1	f_{\rm c}\left(\{\mathbf{w}_{k}\},\phi \right) + \varpi_2 f_{\rm s} \left(\{\mathbf{w}_{k}\},\phi \right) \nonumber \\
		{\rm {s.t.}}~~ & \sum_{k=1}^{K} ||\mathbf{w}_{k}||_2^{2} \le P,\label{C_Pt}\\
		&\phi \in \mathcal{C}_{\phi},\label{C_Ro_area}
	\end{align}
\end{subequations}
where $ \varpi_1 $ and $\varpi_2   $  are the weighting factors for $f_{\rm c}\left(\{\mathbf{w}_{k}\},\phi \right)$ and $ f_{\rm s}\left(\{\mathbf{w}_{k}\},\phi \right) $, respectively, which satisfy $ \varpi_1 + \varpi_2  = 1  $.
Additionally,  \eqref{C_Pt} is the transmit power constraint for the BS, and  \eqref{C_Ro_area} is the array rotation constraint for restricting the admissible rotation region within $\mathcal{C}_{\phi} = [\phi_{\rm min},\phi_{\rm max}]$. 
Note that Problem (\textbf{P1}) is a multi-objective optimization problem, which can balance between C\&S performances by adjusting the weighting factors $\{\varpi_1,\varpi_2\}$. Nevertheless, the solution to Problem (\textbf{P1}) is difficult to be obtained directly due to the non-concave objective function and coupled variables in Problem (\textbf{P1}), namely, $\{\mathbf{w}_{k}\}$ and $\phi$. Furthermore, the array rotation angle constraint makes it more difficult to solve this problem.

To solve Problem (\textbf{P1}), we first equivalently reformulate it by decomposing the problem into two subproblems, corresponding to an inner problem for optimizing the transmit beamforming vectors $\{\mathbf{w}_{k}\}$ given array rotation angle $\phi$, and an outer problem for optimizing the array rotation angle $\phi$, detailed in the following. 

\textbf{\underline{Inner problem}:}  For any feasible array rotation angle $\phi$, the inner problem for transmit beamforming optimization is formulated as 
\begin{align}
		(\textbf{P2}):\; \max_{\{\mathbf{w}_{k}\}}~& \varpi_1	f_{\rm c}\left(\{\mathbf{w}_{k}\},\phi \right) + \varpi_2 f_{\rm s} \left(\{\mathbf{w}_{k}\},\phi\right) \nonumber \\
		{\rm {s.t.}}~~ &\eqref{C_Pt}. \nonumber
\end{align}
We denote  $ g(\phi)\triangleq \varpi_1 f_{\rm c}(\{\bar{\mathbf{w}}_{k}\},\phi) +  \varpi_2 f_{\rm s}(\{\bar{\mathbf{w}}_{k}\},\phi)$, which represents the optimized objective value of Problem (\textbf{P2}), where $\bar{\mathbf{w}}_{k}$ is the optimized beamforming vectors.

\textbf{\underline{Outer problem}:}  Given the inner problem, the outer problem is to optimize  $\phi$ for maximizing $  g(\phi) $, which is given by
\begin{align}
		(\textbf{P3}):\; \max_{\phi}~& g(\phi) \nonumber \\
		{\rm {s.t.}}~~ &\eqref{C_Ro_area}. \nonumber	
\end{align}

\subsection{Proposed Solutions to Problems (\textbf{P2}) and (\textbf{P3})}\label{Sec3-C}
{\emph{(1)~Solution to (\textbf{P2})}}:
Note that for the inner problem (\textbf{P2}), the sensing performance metric (i.e., $ \varpi_2  f_{\rm s} \left(\{\mathbf{w}_{k}\},\phi \right) $) is presented in the objective function; thus, it is not possible to directly apply the weighted minimum mean-square error (WMMSE) method or the closed-form expression obtained by the fractional programming (FP) method~\cite{Guo2020WSR} for the communication-only case to solve it. To address this issue, we extend the FP method to a more generalized form, which incorporates both the fractional and non-fractional terms into a new fractional form. Based on this, the inner problem can be efficiently solved via the BCD technique. 

Specifically, we first introduce auxiliary variables $ \boldsymbol{\alpha} = [\alpha_{1},\alpha_{2},\ldots,\alpha_{K}]^T $ to circumvent the complex fractional forms of SINR within the logarithm function, and impose that  $\alpha_{k}>0, \forall  k\in \{1,2,\ldots,K\} $. As such, the transmit beamforming vectors can be obtained by solving the following problem
\begin{align}
	(\textbf{P2.1}):\; \max_{\{\mathbf{w}_{k}\},\boldsymbol{\alpha}}~&\tilde{f} (\{\mathbf{w}_{k}\},\boldsymbol{\alpha}) \nonumber \\
	{\rm {s.t.}}~~ &~\eqref{C_Pt}. \nonumber
\end{align}
In Problem (\textbf{P2.1}), the objective function is given by 
\begin{align}\label{obj_lag}
	\tilde{f} (\{\mathbf{w}_{k}\},\boldsymbol{\alpha}) &= \sum_{k=1}^{K} \Big(
	\varpi_1 \big( \log_2(1+\alpha_{k}) \big) + \varpi_2 \cos^2\tilde{\theta} \big|\mathbf{a}_{\rm T}^H(\tilde{\theta}) \mathbf{w}_{k}\big|^2 \nonumber \\
	-\varpi_1 \alpha_{k} &+   \varpi_1 (1+\alpha_{k}) \frac{ |\mathbf{h}_{k}^{H} \mathbf{w}_{k}|^2 }{\sum_{i=1}^{K} |\mathbf{h}_{k}^{H} \mathbf{w}_{i}|^2 +\sigma_{k}^{2} } 
	\Big),
\end{align}
which contains both the fractional and non-fractional terms in~\eqref{obj_lag}, making it challenging to optimally solve Problem (\textbf{P2.1}). To address this issue, we regard
 $\varpi_2 \cos^2\tilde{\theta} \sum_{k=1}^{K} \Big|\mathbf{a}_{\rm T}^H(\tilde{\theta}) \mathbf{w}_{k}\Big|^2$ in~\eqref{obj_lag} as a special case of fractional forms. As such, the quadratic transform method~\cite{Shen2018} can be employed to transform~\eqref{obj_lag} into a more tractable form, with the objective function re-expressed as
\begin{align}
	&\hat{f}(\{\mathbf{w}_{k}\},\boldsymbol{\alpha},\mathbf{b},\mathbf{b}_{\rm s}) = \sum_{k=1}^{K} \Big( 
	\varpi_1 \big( \log_2(1+\alpha_{k}) \big)
	-\varpi_1 \alpha_{k} \nonumber \\
	&~~+2\sqrt{\varpi_1 (1+\alpha_{k})} \mathcal{R}\{ b_{k}^{\dagger}  \mathbf{h}_{k}^{H} \mathbf{w}_{k}  \} - | b_{k} |^2 \Big(\sum_{i=1}^{K}  |\mathbf{h}_{k}^{H} \mathbf{w}_{i}|^2 +\sigma_{k}^{2}\Big) \nonumber\\
	&~~+2 \sqrt{\varpi_2 \cos^2\tilde{\theta}} \mathcal{R}\{ b_{{\rm s},k}^{\dagger} \mathbf{a}_{\rm T}^H(\tilde{\theta}) \mathbf{w}_{k}  \}-|b_{{\rm s},k}  |^2
	\Big).
\end{align}
 Herein, $(\cdot)^\dagger$ denotes the conjugate operation. In addition, $ \mathbf{b} = [b_1, b_2, \ldots, b_{K}]^T $ and $ \mathbf{b}_{\rm s} = [b_{{\rm s},1}, b_{{\rm s},2}, \ldots, b_{{\rm s},K}]^T $  are the auxiliary variables,  introduced for addressing the fractional forms (i.e., $ \frac{ |\mathbf{h}_{k}^{H} \mathbf{w}_{k}|^2 }{\sum_{i=1}^{K} |\mathbf{h}_{k}^{H} \mathbf{w}_{i}|^2 +\sigma_{k}^{2} } $) and non-fractional forms (i.e., $\big|\mathbf{a}_{\rm T}^H(\tilde{\theta}) \mathbf{w}_{k}\big|^2$) in~\eqref{obj_lag}, respectively. 
 As $ \hat{f} (\{\mathbf{w}_{k}\},\boldsymbol{\alpha},\mathbf{b},\mathbf{b}_{\rm s}) $ is a concave function w.r.t. each variable in $\{\{\mathbf{w}_{k}\},\boldsymbol{\alpha},\mathbf{b},\mathbf{b}_{\rm s}\}$, we can obtain the optimal $\boldsymbol{\alpha} $, $\mathbf{b} $, $ \mathbf{b}_{\rm s}$ and $ \{\mathbf{w}_{k}\} $ for the transformed problem of (\textbf{P2.1}),  based on their first-order optimality conditions, which are respectively obtained as 
\begin{subequations}\label{BCD_closed}
	\begin{align}
		&\alpha^{*}_{k} = \frac{ |\mathbf{h}_{k}^{H} \mathbf{w}_{k}|^2 }{\sum_{i=1,i\neq k}^{K} |\mathbf{h}_{k}^{H} \mathbf{w}_{i}|^2 +\sigma_{k}^{2} },\forall k \!\in\! \{1,\ldots,K\}, \\ 
		&b^{*}_{k} = \frac{\sqrt{\varpi_1 (1+\alpha_{k})}\mathbf{h}_{k}^{H} \mathbf{w}_{k}}{\sum_{i=1}^{K}  |\mathbf{h}_{k}^{H} \mathbf{w}_{i}|^2 +\sigma_{k}^{2}},\forall k \!\in\! \{1,\ldots,K\}, \\
		&b^{*}_{{\rm s},k} = \sqrt{\varpi_2 \cos^2\tilde{\theta}   }\mathbf{a}_{\rm T}^H(\tilde{\theta}) \mathbf{w}_{k},\forall k \!\in\! \{1,\ldots,K\}, \\
		&\mathbf{w}^{*}_{k}\! =\! \Big(\mu \mathbf{I}_{M_{\rm t}} \!+\! \sum_{i=1}^{K} \!| b_{i} |^2 \mathbf{h}_{i} \mathbf{h}_{i}^{H}  \Big)^{-1} \!{\mathbf{h}}_{{\rm eff},k},\forall k \!\in\! \{1,\ldots,K\}, \label{beamforming18d}
	\end{align}
\end{subequations}
with $ {\mathbf{h}}_{{\rm eff},k} \triangleq b_{k} \sqrt{\varpi_1 (1+\alpha_{k})} \mathbf{h}_{k}  + b_{{\rm s},k} \sqrt{\varpi_2 \cos^2\tilde{\theta} } \mathbf{a}_{\rm T}(\tilde{\theta})  $.  By iteratively updating $\{\boldsymbol{\alpha}, \mathbf{b}, \mathbf{b}_{\rm s}, \{\mathbf{w}_{k}\}\}$, we obtain  suboptimal beamforming vectors $\{\mathbf{w}^{*}_{k}\}$ to Problem (\textbf{P2}). 

{\emph{(2)~Solution to (\textbf{P3})}}:  Given the solution to the inner problem, the outer problem (\textbf{P3}) is still challenging to solve directly due to the lack of a closed-form expression for $\{\mathbf{w}^{*}_{k}\}$. Moreover, the derivatives of the steering vectors (e.g., $ \mathbf{a}_{\rm T}({\tilde{\theta}}) $ and $ \mathbf{a}_{\rm R}({\tilde{\theta}}) $) w.r.t. $\phi$ are  highly complicated, rendering the problem more difficult to solve. To tackle these difficulties, we propose to apply a 1D exhaustive search to find the optimal array rotation angle within the admissible rotation region $\mathcal{C}_{\phi}$.

\subsection{Discussions}\label{Sec3-D}
{\emph{(1)~Algorithm convergence and computational complexity}}: Note that the convergence of the FP-based method has been  verified in~\cite{Guo2020WSR} and~\cite{Shen2018}. Thus, the convergence of the proposed algorithm can be guaranteed.
Next, we consider the computational complexity analysis of the proposed algorithm. For the inner problem, its computational complexity is mainly determined by the FP-based method, which is in the order of $\mathcal{O}(KM_{\rm t}^3)$~\cite{Guo2020WSR}. Hence, the computational complexity order for solving the inner problem is $\mathcal{O}(I_{2} KM_{\rm t}^3)$, where $I_{2}$ represents the number of BCD iterations.
Then, by denoting $ I_{1} $ as the number of 1D exhaustive search, the overall computational complexity for solving Problem (\textbf{P1}) is in the order of $\mathcal{O}(I_{1} I_{2} KM_{\rm t}^3)$.

{\emph{(2)~Special cases}}:
To obtain useful insights into the RA-array in both C\&S, we consider two special cases, namely, the communication-only case (i.e., $\varpi_1 = 1 $ and $\varpi_2 = 0 $) and the sensing-only case (i.e.,  $\varpi_1 = 0 $ and $\varpi_2 = 1 $). Note that the proposed algorithm for solving Problem (\textbf{P1}) can be applied for solving the problems for these two special cases.

\textbf{\emph{Communication-only case}:}
When $ \varpi_1 = 1 $ and $\varpi_2 = 0 $, Problem (\textbf{P1}) reduces to the conventional sum-rate maximization problem. It can be shown that the sum-rate analysis for the multi-user case is intractable for RA systems. Thus, we consider the single-user case (i.e., $K=1$) with an LoS path and a NLoS path, and evaluate the effect of array rotation on the communication performance.
For this case, the BS-user channel reduces to  
\begin{align}\label{simple_channel}
	\mathbf{h}^{H}_{1}(\phi) = 
	\beta_{1,0} \mathbf{a}^{H}(\tilde{\theta}_{1,0} ) + 
	\beta_{1,1} \mathbf{a}^{H}(\tilde{\theta}_{1,1} ) .
\end{align}
As such, the received signal at the user can be expressed as
$	y_{1} = \mathbf{h}_{1}^{H} \mathbf{w}_{1} s_{1} + z_{1},$
where $ z_{1}\sim\mathcal{CN}(0,\sigma_{1}^2)$ is the AWGN.
The optimal transmit beamforming for the user can be obtained based on maximal ratio transmission (MRT) (i.e., $\mathbf{w}_{1} = \sqrt{P} \frac{\mathbf{h}_{1}}{||\mathbf{h}_{1}||_2}$). Then, the corresponding maximum SNR at the user is obtained~as
\begin{align}
	\gamma_{1} &=  \frac{ P|\mathbf{h}_{1}^{H} \mathbf{h}_{1}| }{\sigma_{1}^{2} } \overset{(a)}{=} \left( |\beta_{1,0}|^2 + |\beta_{1,1}|^2\right) M_{\rm t}  \nonumber \\
	&+ 2|\beta_{1,0} \beta_{1,1} | \mathcal{R}\{ \mathbf{a}^H(\tilde{\theta}_{1,0}  )  \mathbf{a}(\tilde{\theta}_{1,1}  )  \}, 
\end{align}  
since $\mathbf{a}^H(\tilde{\theta}_{1,0}  ) \mathbf{a}(\tilde{\theta}_{1,0}  ) = \mathbf{a}^H(\tilde{\theta}_{1,1}  ) \mathbf{a}(\tilde{\theta}_{1,1}  ) = M_{\rm t} $. As such, the rotation angle of RA-array is optimized to maximize the correlation between $\mathbf{a}(\tilde{\theta}_{1,0}  )$ and $ \mathbf{a}(\tilde{\theta}_{1,1}  ) $.
To characterize the communication performance gain introduced by RAs, we define a rotation gain of RAs as follows.
\begin{definition}[Rotation Gain] 
	\rm The rotation gain of RAs  in communication performance is defined as
	\begin{align}\label{Com_phi}
		G_{R}(\phi) =& \frac{\Big|\mathbf{a}^H(\tilde{\theta}_{1,0}  )  \mathbf{a}(\tilde{\theta}_{1,1}  ) \Big|}{\Big|\mathbf{a}^H({\theta}_{1,0}  )  \mathbf{a}({\theta}_{1,1}  ) \Big|}, 
	\end{align}
which is the ratio of the gain of RAs over that of fixed-rotation antennas. 
\end{definition}
Based on the above, by denoting $\zeta_{1} = \frac{\theta_{1,0}-\theta_{1,1}}{2}$ and $ \zeta_{2} =\frac{\theta_{1,0}+\theta_{1,1}}{2} $, the maximum rotation gain for RAs in communication performance is obtained as
\begin{align}
	G_{R}(\phi^{*}) =	M_{\rm t}{\Bigg|
		\frac{\sin\big(\pi \sin\zeta_{1} \cos\zeta_{2} \big)}
		{\sin\big( M_{\rm t } \pi  \sin\zeta_{1} \cos\zeta_{2}\big)}
		\bigg|},
\end{align}
where the optimal array rotation is set as $ \phi^{*} + \zeta_{2} = \frac{n\pi}{2} $, with $n$ being any non-zero integer. Note that $ G_{R}(\phi) > 1$  always holds when $ \theta_{1,0} \neq \theta_{1,1} $. 
Thus, adjusting the rotation angle of RAs helps to increase the correlation between $\mathbf{a}(\tilde{\theta}_{1,0}  )$ and $ \mathbf{a}(\tilde{\theta}_{1,1}  ) $, leading to the enhanced received SNR.

\textbf{\emph{Sensing-only case}:}
When $ \varpi_1 = 0 $ and $\varpi_2 = 1 $, Problem (\textbf{P1}) reduces to a quadratic programming problem, which aims to maximize $f_{\rm s} \left(\{\mathbf{w}_{k}\},\phi \right) $ (equivalently minimizing the CRB for angle estimation). The optimal beamforming matrix for this problem can be obtained as ${\mathbf{W}^{*}} =\frac{P}{M_{\rm t} }  \mathbf{a}_{\rm T}(\tilde{\theta}) \mathbf{a}_{\rm T}^H(\tilde{\theta}) $, based on which, we have $  \mathbf{a}_{\rm T}^H(\tilde{\theta}) \mathbf{W}^{*} \mathbf{a}_{\rm T}(\tilde{\theta}) =P M_{\rm t} $. Therefore, to maximize $f_{\rm s} \left(\{\mathbf{w}_{k}\},\phi \right) $, the rotation angle $\phi$ should be optimized to maximize $\cos\tilde{\theta}$, which can be obtained by solving the following problem $\max_{\phi\in \mathcal{C}_{\phi}}~\cos^2 (\theta+\phi)$. It can be easily shown that the optimal solution satisfies $\phi^{*}+\theta = 0~\text{or}~\pi$.

\begin{remark}[Equivalent spatial aperture]\label{SEA}
	\rm
	For the sensing-only case, the CRB for target angle estimation is determined by the \emph{equivalent spatial aperture}, which is defined as $D_{\rm SEA} = D_{\rm r} \cos(\theta+\phi)$. For  fixed-rotation antennas, $ D_{\rm SEA} $ is determined by the target direction $\theta$, while RAs offer an additional DoF for attaining a lower CRB by adjusting its array rotation angle $\phi$, which will be  numerically validated in Section~\ref{Sec:SR}. 
\end{remark}

\section{Numerical Results}\label{Sec:SR}

Numerical results are presented in this section to validate the efficiency of RAs for ISAC. Unless otherwise specified, the system parameters are set in Table~\ref{SimulationParameters}. 
\begin{table}[t!]
	\centering
	\caption{System Parameters}
	\label{SimulationParameters}
	\begin{tabular}{c|c}
		\hline 
		\textbf{Parameter} & \textbf{Value}          
		\\ \hline Number of transmit antennas $M_{\rm t}$ & 16
		\\ \hline Number of receive antennas $M_{\rm r}$ & 16
		\\ \hline Number of NLoS paths $L$  & 5
		\\ \hline Number of users $K$  & 4
		\\ \hline Maximum transmit power $P$  & 1 W
		\\ \hline Number of snapshots $T$ & 100
		\\ \hline Distributions of $\theta$ and $\theta_{k,\ell}$ & $\mathcal{U}(-\frac{\pi}{3},\frac{\pi}{3})$
		\\ \hline Admissible rotation region $\mathcal{C}_{\phi}$ & $ [-\frac{\pi}{6},\frac{\pi}{6}]$
		\\ \hline Normalized path gain $ {|\beta_{k,\ell}|^2}/{\sigma_k^2} $ & 0 dB 
		\\ \hline $\text{SNR}_{\rm s} {|\beta_{\rm s}|^2}/{\sigma_{\rm s}^2} $ & -10 dB
		\\ \hline
	\end{tabular}
	\vspace{-10pt}
\end{table}
\begin{figure}[t]
	\centering
	\includegraphics[width=0.35\textwidth]{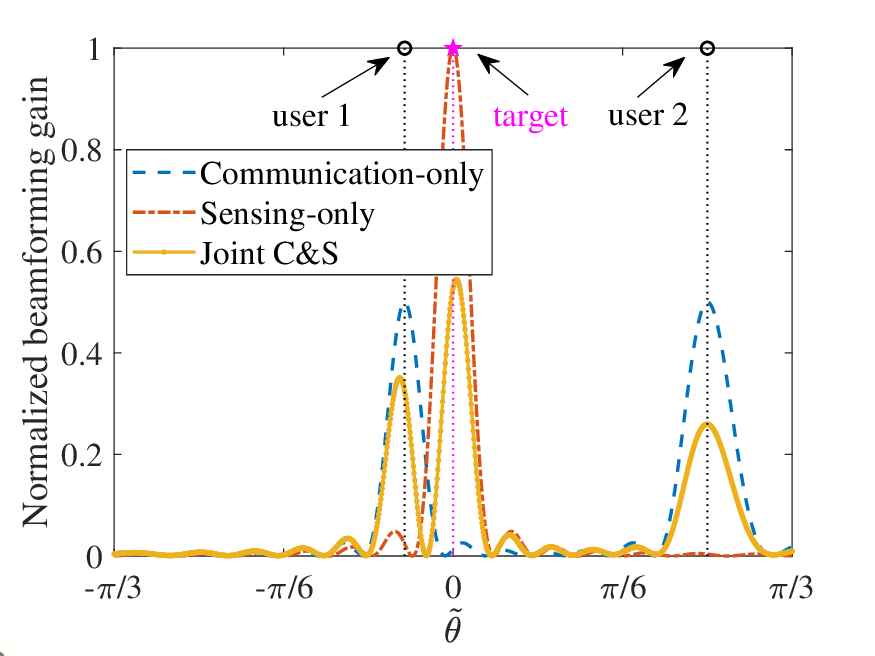}
	\caption{Normalized beamforming gain for different cases.} 
	\label{Fig:Waveforms}
	\vspace{-10pt}
\end{figure}
For performance comparison, we consider the following benchmark schemes:
\begin{itemize}
	\item \emph{Beamforming optimization only}: In this scheme, the array rotation angle is fixed, while the transmit beamforming vectors are optimized  by using the proposed method.
	
	\item \emph{Rotation optimization only}: In this scheme, the array rotation angle is optimized by using the proposed method, while the transmit beamforming vectors are designed via the zero forcing (ZF) technique.
	
\end{itemize}

To showcase the efficiency of the proposed scheme in balancing C\&S performances, the normalized beamforming gain, denoted by $\text{Tr}(\mathbf{a}^H_{\rm T}({\tilde{\theta}}) \mathbf{W} \mathbf{a}_{\rm T}({\tilde{\theta}}) ) $, across the effective spatial angle ranging from $-\frac{\pi}{3}$ to $\frac{\pi}{3}$ is presented in Fig.~\ref{Fig:Waveforms} with $K=2$. From Fig.~\ref{Fig:Waveforms}, it is observed that, for both sensing-only or communication-only cases, the transmit beam is steered towards the direction of the target or users to minimize the CRB or maximize achievable sum-rate, whereas for the joint C\&S case,  waveforms are optimized to concurrently perform downlink communication and target sensing. However, the normalized beamforming gain is slightly reduced, indicating that there exists a performance trade-off between~C\&S. 

Fig.~\ref{Fig:Tradeoff} shows the achievable  C\&S performances under different schemes. Note that, in contrast to the scheme without rotation angle design (i.e., beamforming optimization only), the proposed RA scheme consistently attains a lower CRB and/or a higher achievable sum-rate, thereby demonstrating the efficiency of RAs in enhancing both C\&S performances. In addition, compared with the rotation optimization only scheme, the proposed algorithm is capable of attaining a lower CRB for target sensing or a higher achievable sum-rate, which verifies the superior efficiency of the proposed joint design. 

\begin{figure}[t]
	\centering
	\includegraphics[width=0.35\textwidth]{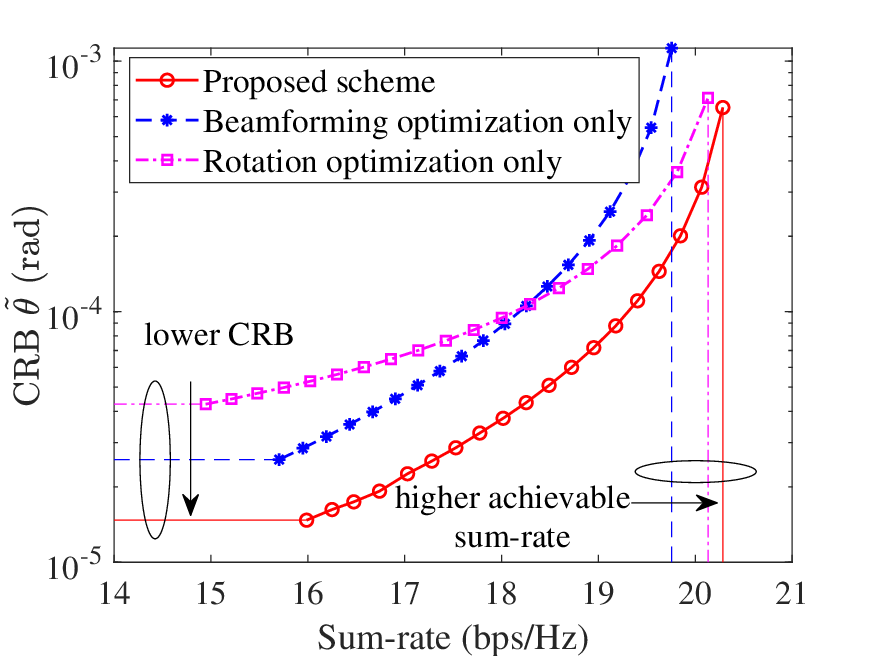}
	\caption{Achievable C\&S performances for different schemes.} 
	\label{Fig:Tradeoff}
	\vspace{-10pt}
\end{figure}

\section{Conclusions}\label{Sec:Con}
In this letter, we proposed an RA-enabled ISAC system, where an RA-array is deployed at the BS to enhance both C\&S performances by exploiting an additional spatial DoF offered by the array rotation. Specifically, we formulated an optimization problem to maximize weighted C\&S performances by jointly designing the transmit beamforming and the array rotation angle. To solve this non-convex problem, we transformed it into two subproblems, and then applied the BCD technique and 1D exhaustive search to solve them efficiently, respectively. In addition, to obtain useful insights, two special cases were considered to showcase the superiority of RAs in improving the rotation gain for the enhanced communication SNR and enlarging the equivalent spatial aperture for the more accurate target sensing. Finally, we presented numerical results to demonstrate the superior performance of RAs over fixed-rotation antennas for ISAC. 

\bibliographystyle{IEEEtran}
\bibliography{Ref_RAISAC.bib}

\end{document}